\documentclass[12pt,a4paper]{amsart}
\usepackage[a4paper, footskip=0.5in,  headheight = 0.5in, top=1.25in, bottom=1.25in,  right=1in,  left=1in]{geometry}
\usepackage{macros}

\title[Clique-width and induced topological minors]{Clique-width and induced topological minors}

\author{Paweł Rafał Bieliński$^{\dagger}$}
\thanks{$^{\dagger}$Warsaw University of Technology, Poland. Supported by the National Science Centre grant 2024/54/E/ST6/00094.}
\author{Jadwiga Czyżewska$^{\ddagger}$}\thanks{$^{\ddagger}$University of Warsaw, Poland (\texttt{j.czyzewska@mimuw.edu.pl}).
Supported by Polish National Science Centre SONATA BIS-12 grant number 2022/46/E/ST6/00143.}
\author{Martin Milanič$^{\mathsection}$}\thanks{$^{\mathsection}$FAMNIT and IAM, University of Primorska, Slovenia (\texttt{martin.milanic@upr.si}). 
Supported in part by the Slovenian Research and Innovation Agency (I0-0035, research program P1-0285 and research projects J1-60012, J1-70035, J1-70046, and N1-0370) and by the research program CogniCom (0013103) at the University of Primorska.}
\author{Amir Nikabadi$^{\ast}$}\thanks{$^{\ast}$IT University of Copenhagen, Denmark (\texttt{amir@itu.dk}).
Supported by the Independent Research Fund Denmark (DFF), grant agreement number 2098-00012B}
\author{Paweł Rzążewski$^{\parallel}$}\thanks{$^{\parallel}$Warsaw University of Technology, Poland (\texttt{pawel.rzazewski@pw.edu.pl}). Supported by the National Science Centre grant 2024/54/E/ST6/00094.}

\begin{document}

\maketitle

\begin{abstract}
A $P_4$ is a chordless path on four vertices. 
A \emph{diamond} is a~graph obtained from a~clique of size four by removing one edge of the clique.
A \emph{paw} is a~graph obtained from a~clique of size four by removing
two adjacent edges of the clique.
We prove that for a graph $H$, the class of graphs with no induced subdivision of $H$ has
bounded clique-width if and only if $H$ is an induced subgraph of $P_4$, the paw, or the diamond. 
This answers a~question of Dabrowski, Johnson, and Paulusma.
\end{abstract}

\section{Introduction}
Graph classes defined by excluding a~fixed graph with respect to some containment
relation form a~central topic in structural graph theory. 
For the minor relation, the Grid-Minor Theorem of Robertson and Seymour~\cite{robertson1986graph} implies a~complete characterization of the graphs $H$ for which the class of $H$-minor-free graphs
has bounded treewidth: this happens precisely when $H$ is planar. 
Because of the pivotal role of treewidth in algorithmic graph theory, there have been numerous attempts, both successful and unsuccessful, to find similar theorems for other containment relations combined with an appropriate width parameter. 
The containment relation vertex-minors combined with rank-width~\cite{geelen2023grid},  
immersions combined with tree-cut width~\cite{wollan2015structure}, and induced minors combined with tree-independence number~\cite{dallard2024treewidth2} are such examples. 

On the other hand, generalizing the Grid-Minor Theorem to induced subgraphs is an active area of research. 
This line of work is developed in the series ``Induced subgraphs and tree decompositions'', whose stated goal is to understand which induced subgraphs are unavoidable in graphs of large treewidth~\cite{abrishami2022induced,chudnovsky2025complete,chudnovsky2025thetas}.

In this paper, we focus on the setting in which the containment relation is the \emph{induced topological minor} relation and the width parameter is \emph{clique-width}.
Clique-width, denoted by $\cw$, has been studied extensively both in algorithmic and structural graph theory. 
It generalizes the aforementioned classical notion of \emph{treewidth}, denoted by $\tw$, in the following sense: For every graph~$G$, it holds that ${\cw(G)= \mathcal{O}(2^{\tw(G)})}$~\cite{corneil2005relationship}, whereas no converse bound holds, since there are graphs of arbitrarily large treewidth and bounded clique-width (a graph class $\mathcal{G}$ has bounded clique-width if there exists $c \in \mathbb{N}$ such that $\cw(G) \leq c$ for all $G \in \mathcal{G}$).

\medskip

Dabrowski, Johnson, and Paulusma~\cite{Dabrowski_Johnson_Paulusma_2019} posed the following
problem: 

\begin{problem}[{Dabrowski, Johnson, and Paulusma~\cite[Problem~4.34]{Dabrowski_Johnson_Paulusma_2019}}]\label{openprob}
For which graphs $H$ does the class of $H$-induced-topological-minor-free graphs have bounded clique-width?
\end{problem}

We give a~complete characterization of classes of $H$-induced-topological-minor-free graphs of bounded clique-width, and thereby provide a~solution to Problem~\ref{openprob}.
In particular, we prove the following:

\begin{restatable}{theorem}{thmmain}\label{thm:main}
Let $H$ be a~graph.
Then, the class of $H$-induced-topological-minor-free graphs has bounded clique-width if and only if $H$ is an induced subgraph of the $P_4$, the paw, or the diamond.
\end{restatable}

The proof of Theorem~\ref{thm:main} combines known characterizations for $H$-free graphs and $H$-induced-topological-minor-free graphs for particular graphs $H$, as well as the analogue of Theorem~\ref{thm:main} for $H$-induced-minor-free graphs.

Throughout the note, graphs have finite vertex sets, no loops, and no parallel edges. 
We use $\mathbb{N}$ to denote the set of positive integers. We let $[n]\coloneqq\{1,\dots, n\}$ for every $n\in \mathbb{N}$. 
Let $G = (V,E)$ be a~graph. A \textit{clique} in $G$ is a~set of pairwise adjacent vertices. 
For $X \subseteq V(G)$, we denote the subgraph of $G$ induced by $X$ as $G[X]$, that is, $G[X] = (X, \{uv \colon u, v \in X \mbox{ and } uv \in E \})$.
We denote by $N(X)$ vertices of $G\setminus X$ with a~neighbor in $X$, and $N[X] = N(X)\cup X$.
We denote by $\Delta(G)$ the maximum degree of~$G$, that is, the maximum value of $|N(\{v\})|$ over all $v\in V$.
A graph is \emph{subcubic} if its maximum degree is at most $3$ and  \emph{complete multipartite} if it admits a partition of its vertex set such that two distinct vertices are adjacent if and only if they belong to distinct parts.

Given two graphs $G$ and $H$, the graph $H$ is a \emph{subdivision} of $G$ if $H$ can be obtained from $G$ by a sequence of edge subdivisions, a \emph{topological minor} of $G$ if some subdivision of $H$ is a subgraph of $G$, an \emph{induced minor} of $G$ if $H$ can be obtained from $G$ by a~sequence of vertex deletions or edge contractions, 
and an \emph{induced topological minor} of $G$ if some subdivision of $H$ is an induced subgraph of $G$
(in which case we also say that $G$ \emph{contains an induced subdivision of $H$}).
If $H$ can be obtained from $G$ by a~sequence of vertex deletions, edge deletions, and edge contractions, then $H$ is said to be a~\emph{minor} of $G$.
If $G$ does not contain an induced subgraph isomorphic to $H$, 
then we say that $G$ is \emph{$H$-free}.
Analogously, we say that $G$ is \emph{$H$-topological-minor-free}, \emph{$H$-induced-topological-minor-free}, \emph{$H$-minor-free}, or \emph{$H$-induced-minor-free}, respectively.
Similarly, we say that $G$ is \emph{$\mathcal H$-free} for a set $\mathcal{H}$ of graphs if $G$ is $H$-free for all $H\in \mathcal H$.

  \begin{figure}[ht]
        \centering
        \includegraphics[width=\textwidth]{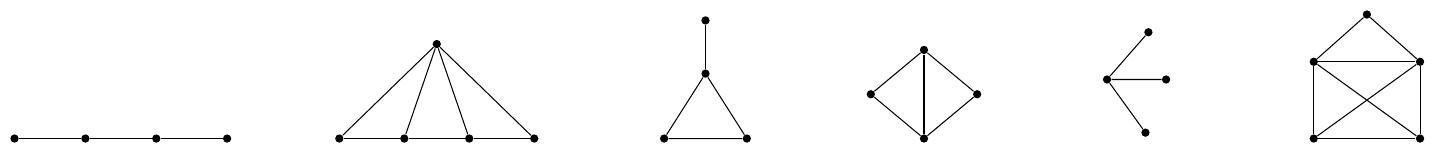}
        \caption{From left to right: the $P_4$, the gem, the paw, the diamond, the claw, and the fat-house.}
        \label{fig:graphs}
    \end{figure}

We let $P_n$ and $K_n$ denote the chordless path and the complete graph on $n$ vertices. 
A \emph{diamond} is a~graph obtained from a~$K_4$ by removing one edge of the $K_4$. 
A \emph{paw} is a~graph obtained from a~$K_4$ by removing two adjacent edges of the $K_4$.
A \emph{claw} is a~graph obtained from a~paw by removing the edge between the two vertices of degree $2$.
A \emph{fat-house} is a~graph obtained from a~$K_4$ by adding a~vertex adjacent to exactly two vertices of the~$K_4$.
A \emph{gem} is graph obtained from a~$P_4$ by adding a~vertex 
adjacent to all vertices of the $P_4$.
(See Figure~\ref{fig:graphs}).
The \textit{$(n \times m)$-grid}, denoted $G_{n\times m}$, is the graph with vertex set $[n] \times [m]$ and edge set
\[    
\{ \{ (i, j), (i, j+1) \} \colon i \in [n], j \in [m-1] \} \cup \{ \{ (i, j), (i+1, j) \} \colon i \in [n-1], j \in [m] \}.
\]
The \textit{($k\times k$)-wall} for $k\geq 3$, denoted $W_{k\times k}$, is the graph $G$ obtained from the $(k \times 2k)$-grid $G_{k \times 2k}$ by deleting every
odd edge in every odd column and every even edge in every even column, and then deleting all degree-one vertices (see Figure~\ref{fig:grid-wall}). 

 \begin{figure}[ht]
    \centering
\includegraphics[width=0.55\linewidth]{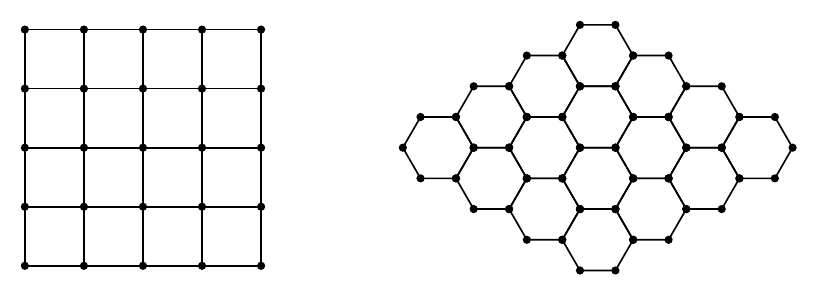}
    \caption{The 5-by-5 square grid (left) and the 5-by-5 wall $W_{5\times 5}$ (right).}
    \label{fig:grid-wall}
\end{figure}

We will make use of the following fact about unboundedness of clique-width in the class of walls and their line graphs (we do not need the definition of clique-width in this paper).

\begin{observation}\label{obs:cw}
    The class of walls has unbounded clique-width.
    The class of line graphs of walls has unbounded clique-width.
\end{observation}

We briefly justify Observation~\ref{obs:cw}. The first assertion follows from combining two facts: (i) treewidth is unbounded in the class of  walls~(see Reed~\cite{reed1997tree}), which are planar graphs of maximum degree at most 3, and (ii) there exists a function $f$ such that every graph $G$ satisfies $\tw(G)\leq f(\Delta(G),\cw(G))$ (see, e.g.,~\cite{kaminski2009recent}).
The second assertion follows from a~result of Gurski and Wanke~\cite{gurski2007line} stating that a~class of graphs has bounded treewidth if and only if the class of their line graphs has bounded clique-width. 
Since walls have unbounded treewidth, their line graphs have unbounded clique-width.


\section{The Proof}
We need a~couple of results from the literature, beginning with the characterizations of classes of $H$-free graphs of bounded clique-width:

\begin{lemma}[{Dabrowski and Paulusma~\cite[Theorem~6]{dabrowski2016clique}}]\label{lem:p4}
Let $H$ be a~graph.
Then, the class of $H$-free graphs has bounded clique-width if and only if $H$ is an induced subgraph of $P_4$.
\end{lemma}

We also need the analogous result for $H$-induced-minor-free graphs:

\begin{lemma}[{Belmonte, Otachi, and Schweitzer~\cite[Theorem~1.2]{belmonte2018induced}}]\label{lem:gem-fathouse}
Let $H$ be a graph. The class of $H$-induced-minor-free graphs has bounded
clique-width if and only if $H$ is an induced subgraph of the fat-house or of the gem.
\end{lemma}

Next, we need the characterization of paw-induced-topological-minor-free graphs:

\begin{lemma}[{Chudnovsky, Penev, Scott, and Trotignon;~\cite[2.1]{chudnovsky2012excluding}}]\label{lem:paw}
A connected graph $G$ is paw-induced-topological-minor-free if and only if $G$ is either a~tree, a~cycle, or a~complete multipartite graph.
\end{lemma}

We also need a~similar result  from~\cite{dallard2021treewidth}.
A \emph{block} of a~graph is a~maximal connected subgraph without cut-vertices.
A \emph{block-cactus graph} is a~graph every block of which is a~cycle or a~complete graph.

    \begin{lemma}[{Dallard, Milanič, and Štorgel~\cite[part of Lemma~3.2]{dallard2021treewidth}}]\label{lem:diam}
A graph $G$ is diamond-induced-topological-minor-free if and only if $G$ is a~block-cactus graph.
\end{lemma}

We also need the following result from~\cite{BoliacLozin2002} relating clique-width of a graph to clique-width its blocks.

\begin{lemma}[Boliac and Lozin~\cite{BoliacLozin2002}]\label{lem:boliac-lozin}
For every graph $G$,
$$\cw(G) \leq \max\{\cw(B)\colon  B \text{ is a block of } G\}+2.$$
Consequently, for every graph class $\mathcal{C}$, if the blocks of all graphs in $\mathcal{C}$ have bounded clique-width, then $\mathcal{C}$ has bounded clique-width.
\end{lemma}

And finally we need the following:

\begin{lemma}[{Munaro~\cite[part of Theorem 4]{munaro2017line}}]\label{lem:line-subcubic}
A graph $G$ is $\{K_4$, diamond, claw$\}$-free if and only if $G$ is the  line graph of a subcubic $K_3$-free graph.
\end{lemma}
 
We are now ready to prove our main result:

\thmmain*

\begin{proof}
Let $\mathcal X$ be the class of $H$-induced-topological-minor-free graphs.
We first prove the ``only if'' implication. Suppose that $\mathcal X$ has bounded clique-width. Observe that every induced
subdivision of $H$ yields $H$ as an induced minor after contracting the subdivided
edges, so the class of $H$-induced-minor-free graphs is contained in $\mathcal X$.
Hence, by Lemma~\ref{lem:gem-fathouse}, $H$ is an induced subgraph of the gem or of the fat-house. Moreover, $\Delta(H)\le 3$; otherwise no subcubic graph could contain $H$ as a~topological minor, implying that all walls would belong to $\mathcal X$, contradicting Observation~\ref{obs:cw}.
Since both the gem and the fat-house have a vertex of degree~$4$, it follows that $H$ is a proper induced subgraph of the gem or of the fat-house.
If $H$ is a proper induced subgraph of the gem, then a direct inspection shows
that $H$ is an induced subgraph of the $P_4$, the paw, or the diamond.
It remains to discuss the case where $H$ is a proper induced subgraph of the fat-house. 
If $H$ does not contain a clique of size $4$, then by inspecting the fat-house we infer that $H$ is an induced subgraph of the paw or of the diamond. 
Hence, we may assume that $H$ contains a clique of size $4$.
In this case $H$ is isomorphic to the $K_4$.
However, it can be observed that the class $\mathcal{X}$ of $K_4$-induced-topological-minor-free graphs contains the line graphs of walls. Indeed, since every subdivision of the $K_4$ that is not the $K_4$ itself contains either an induced diamond or an induced claw, the class $\mathcal{X}$ contains the class of $\{K_4$, diamond, claw$\}$-free graphs, which, by Lemma~\ref{lem:line-subcubic} coincides with the class of line graphs of subcubic $K_3$-free graphs and, hence, contains the line graphs of walls.
By Observation~\ref{obs:cw}, the class of line graphs of walls has unbounded clique-width, which contradicts the clique-width boundedness of $\mathcal{X}$. 
This proves the only-if direction.

To prove the ``if'' implication, we need to show that if $H$ is an induced subgraph of the $P_4$, the paw, or the diamond,
then $\mathcal{X}$ has bounded clique-width. We use the following simple monotonicity observation: if $H$ is an
induced subgraph of a~graph $F$, then every graph containing $F$ as an induced topological minor also contains $H$ as an induced topological minor. 
Hence, the class of $H$-induced-topological-minor-free graphs is contained in the class of $F$-induced-topological-minor-free graphs.

If $H$ is an induced subgraph of $P_4$, then every graph in $\mathcal{X}$ is $P_4$-free.
Hence $\mathcal{X}$ has bounded clique-width by Lemma~\ref{lem:p4}.

If $H$ is an induced subgraph of the paw, then $\mathcal{X}$ is contained in the class of paw-induced-topological-minor-free graphs. 
By Lemma~\ref{lem:paw}, each block of a~graph in $\mathcal{X}$ is a~cycle or a~complete multipartite graph. 
Now the class of cycles has bounded treewidth and, hence, bounded clique-width, while the class of complete multipartite graphs is contained in the class of $P_4$-free graphs, which has again bounded clique-width by Lemma~\ref{lem:p4}.
Consequently, by Lemma~\ref{lem:boliac-lozin}, $\mathcal{X}$ has bounded clique-width.

Finally, if $H$ is an induced subgraph of the diamond, then $\mathcal{X}$ is
contained in the class of diamond-induced-topological-minor-free graphs. 
By Lemma~\ref{lem:diam}, these graphs are block-cactus graphs. 
Thus every block is a cycle or a~complete graph. 
Since the class of cycles has bounded clique-width and the class of complete graphs is contained in the class of $P_4$-free graphs, which has bounded clique-width by Lemma~\ref{lem:p4}, we infer that every block of any graph in the class of diamond-induced-topological-minor-free graphs has bounded clique-width. 
Now it follows by Lemma~\ref{lem:boliac-lozin} that the class of diamond-induced-topological-minor-free graphs has bounded clique-width.
Hence, again $\mathcal X$ has bounded clique-width. 
This completes the proof of Theorem~\ref{thm:main}.
\end{proof}

\bibliographystyle{plainurl}
\bibliography{ref}

\end{document}